\documentclass[12pt,a4paper,english]{article}
\usepackage[ansinew]{inputenc}
\usepackage[T1]{fontenc}
\usepackage[english]{babel}
\usepackage{tabularx}
\usepackage{fancybox}
\usepackage{fancyhdr}
\usepackage{color}
\usepackage{authblk}
\usepackage{setspace}
\usepackage[numbers, sort]{natbib}

\bibpunct{[}{]}{,}{n}{,}{,}
\usepackage{amsmath}
\usepackage{amstext}
\usepackage{amssymb}
\usepackage{amsthm}
\usepackage{mathrsfs}
\newcommand{\lsi}{\left[\negthinspace\left[}
\newcommand{\rsi}{\right]\negthinspace\right]}
\newcommand{\lsir}{\left(\negthinspace\left(}

\newcommand{\M}{\mathcal{M}}
\newcommand{\Mloc}{\mathcal{M}_{\text{loc}}}
\newcommand{\mMloc}{\mathcal{M}_{\emph{loc}}}

\newcommand{\A}{\mathcal{A}}
\newcommand{\F}{\mathcal{F}}

\newcommand{\FF}{\mathbb{F}}
\newcommand{\GG}{\mathbb{G}}
\newcommand{\R}{I\!\! R}
\newcommand{\N}{I\!\!N}
\newcommand{\EE}{\mathbb{E}}
\newcommand{\PP}{\mathbb{P}}
\newcommand{\QQ}{\mathbb{Q}}

\newcommand{\hatK}{\widehat{K}}

\newcommand{\barA}{\bar{A}}
\newcommand{\barM}{\bar{M}}
\newcommand{\ind}{\mathbf{1}}

\newtheorem{theorem}{\bf Theorem}[section]
\newtheorem{proposition}[theorem]{\bf Proposition}
\newtheorem{lemma}[theorem]{\bf Lemma}

\theoremstyle{remark}
\newtheorem{remark}{\bf Remark}[section]
\newtheorem{definition}[theorem]{\bf Definition}
\newtheorem{example}{\bf Example}[section]
\numberwithin{equation}{section}
\usepackage{geometry}
\makeatletter
\renewcommand*\@fnsymbol[1]{\the#1}
\makeatother

\title{No-arbitrage conditions \\ and absolutely continuous changes of measure}

\author{\vspace{0.2cm} Claudio Fontana}
\affil{\normalsize{Laboratoire Analyse et Probabilit\'e,\\
Universit\'e d'\'Evry Val d'Essonne,\\
23 bd de France, 91037, \'Evry (France).} \\\vspace{0.2cm}
E-mail: \texttt{claudio.fontana@univ-evry.fr}}

\date{This version: March 4, 2014}

\begin{document}

\maketitle

\abstract{\begin{spacing}{1}\noindent 
We study the stability of several no-arbitrage conditions with respect to absolutely continuous, but not necessarily equivalent, changes of measure. We first consider models based on continuous semimartingales and show that no-arbitrage conditions weaker than NA and NFLVR are always stable. Then, in the context of general semimartingale models, we show that an absolutely continuous change of measure does never introduce arbitrages of the first kind as long as the change of measure density process can reach zero only continuously.
\end{spacing}}
\vspace{0.5cm}

\section{Introduction}	\label{S1}

The notion of no-arbitrage is of paramount importance in financial economics and mathematical finance. The search for an adequate no-arbitrage condition, economically meaningful and mathematically tractable, has a rather long history and we refer to \citet{S} for an overview of the main developments. An important step was marked by the papers \citet{DS94,DS98}, where the authors established the equivalence between the existence of an Equivalent (Local-/$\sigma$-) Martingale Measure and the validity of the No Free Lunch with Vanishing Risk (NFLVR) condition.

More recently, motivated mainly by developments in Stochastic Portfolio Theory (see e.g. \citet{FK} for an overview) and in the Benchmark Approach to Quantitative Finance (see e.g. \citet{PH}), an increasing attention has been paid to no-arbitrage concepts weaker than the classical NFLVR condition, as documented by the recent papers \citet{KK}, \citet{HS}, \citet{Ka1,Ka2}, \citet{FR}, \citet{Ruf}, \citet{Song} and \citet{T} (see also \citet{LW} for an earlier contribution to this strand of literature and \citet{H} and \citet{F} for a unifying analysis of the different no-arbitrage conditions appearing in the context of continuous semimartingale models). In particular, it has been demonstrated that the full strength of NFLVR may not be needed in order to solve portfolio optimisation problems as well as to perform pricing and hedging.  

In quantitative finance, the typical approach starts with the specification of a model under a given probability measure $\PP$, under the assumption that arbitrage profits (according to one of the many possible definitions) cannot be realised. One may then wonder whether such a no-arbitrage assumption is robust with respect to a change of measure from $\PP$ to an absolutely continuous, but not necessarily equivalent, measure $\QQ$. The present paper aims at answering this question.

In the mathematical finance literature, it has already been shown that absolutely continuous, but not equivalent, changes of measure may lead to arbitrage opportunities (in the classical sense of \citet{DS94}), see e.g. \citet{DS95a}, \citet{OR} and, more recently, \citet{RR} and \citet{CT}. However, this still leaves open the question of whether no-arbitrage conditions weaker than NA and NFLVR can be altered by absolutely continuous changes of measure. As we are going to show in Section \ref{S3}, the first (but  actually not very surprising) main result of this paper is that this is never the case in models based on continuous semimartingales.

Among the no-arbitrage conditions that are weaker than NFLVR, the condition of No Arbitrage of the First Kind (NA1), or, equivalently, No Unbounded Profit with Bounded Risk (see \citet{KK} and \citet{Ka1}), plays a particularly important role. Indeed, it has been shown that pricing and hedging can be satisfactorily performed as long as NA1 holds and, moreover, the same condition is the minimal one in order to solve portfolio optimisation problems in general semimartingale models. This motivates the study of the stability of NA1 with respect to absolutely continuous, but not necessarily equivalent, changes of measure in general semimartingale models. As already observed in \citet{KK} (see their Remark 4.13), the NA1 condition may not be preserved by an absolutely continuous change of measure when jumps are present. In that regard, the second main result of this paper shows that the loss of NA1 after a change of measure is intimately linked to the possibility of the density process jumping to zero (see Section \ref{S5}).

The remainder of the paper is structured as follows. Section \ref{S2} introduces the modeling framework for the first part of the paper, based on continuous semimartingales, and presents five different notions of arbitrage. In the continuous semimartingale setting, Section \ref{S3} shows that weak no-arbitrage conditions are invariant with respect to absolutely continuous changes of measure, while Section \ref{S4} shows by means of a counterexample that this is not the case for the classical NA and NFLVR conditions. Finally, Section \ref{S5} studies the stability of the NA1 condition in general semimartingale models.

\section{Continuous financial markets and no-arbitrage conditions} \label{S2}

Let $(\Omega,\F,\PP)$ be a given probability space endowed with a right-continuous filtration $\FF=(\F_t)_{0\leq t \leq T}$, where $T\in\left(0,\infty\right)$ represents a fixed time horizon. We denote by $\M(\PP)$ the family of all (uniformly integrable) $\PP$-martingales and by $\Mloc(\PP)$ the family of all $\PP$-local martingales. Without loss of generality, we assume that all elements of $\Mloc$ are c\`adl\`ag and we denote by $\M^c(\PP)$ and $\Mloc^c(\PP)$ the families of all continuous processes belonging to $\M(\PP)$ and $\Mloc(\PP)$, respectively.

We consider an abstract financial market with $d$ risky assets, whose discounted prices (with respect to a given reference security) are represented by the $\R^d$-valued \emph{continuous} semimartingale $S=(S_t)_{0\leq t \leq T}$, with $S_t=(S^1_t,\ldots,S^d_t)^{\top}$, with $^{\top}$ denoting transposition. Being $S$ a special semimartingale, its unique canonical decomposition can be written as $S=S_0+A+M$, where $A$ is a continuous $\R^d$-valued predictable process of finite variation and $M$ is an $\R^d$-valued process in $\Mloc^c(\PP)$ with $M_0=A_0=0$. By Proposition II.2.9 of \citet{JS}, we can write, for all $i,j=1,\ldots,d$,
\begin{equation}	\label{char}
A^i = \int\!a^idB
\qquad\text{and}\qquad
\langle S^i,S^j\rangle = \langle M^i,M^j\rangle = \int\!c^{ij} dB,
\end{equation}
for some continuous real-valued predictable strictly increasing process $B$ and where $a=(a^1,\ldots,a^d)^{\!\top}$ and $c=\bigl((c^{i1})_{1\leq i \leq d},\ldots,(c^{id})_{1\leq i \leq d}\bigr)$ are predictable processes taking values in $\R^d$ and in the cone of symmetric nonnegative $d\times d$ matrices, respectively. We do not necessarily assume that $S$ takes values in the positive orthant of $\R^d$. For all $t\in[0,T]$, let us denote by $c^+_t$ the Moore-Penrose pseudoinverse of the matrix $c_t$. The proof of Proposition 2.1 of \citet{DSpr} (see also \citet{DShi}, Lemma 4.3) shows that the process $c^+=(c^+_t)_{0\leq t \leq T}$ is predictable and, hence, the process $a$ can be represented as
\begin{equation}	\label{char-2}
a = c\,\lambda+\nu,
\end{equation}
where $\lambda=(\lambda_t)_{0\leq t \leq T}$ is defined by $\lambda_t:=c^+_ta_t$, for all $t\in[0,T]$, and $\nu=(\nu_t)_{0\leq t \leq T}$ is an $\R^d$-valued predictable process with $\nu_t\in\text{Ker}(c_t):=\{x\in\R^d:c_tx=0\}$, for all $t\in[0,T]$.

Let us now introduce the notion of \emph{admissible strategy}, assuming a frictionless financial market. Let $L(S;\PP):=L^2_{\text{loc}}(M;\PP)\cap L^0(A;\PP)$, where $L^2_{\text{loc}}(M;\PP)$ and $L^0(A;\PP)$ are the sets of all $\R^d$-valued predictable processes $H$ such that $\int_0^T\!H^{\top}_td\langle M,M\rangle_tH_t<\infty$ $\PP$-a.s. and $\int_0^T\!|H_t^{\top}dA_t|<\infty$ $\PP$-a.s., respectively. It is well-known that $L(S;\PP)$ is the largest set of predictable integrands with respect to the continuous semimartingale $S$ under the probability measure $\PP$. For $H\in L(S;\PP)$, we denote by $H\cdot S$ the stochastic integral process $\bigl(\int_0^t\!H_u\,dS_u\bigr)_{0\leq t \leq T}$, which has to be understood as a vector stochastic integral (see e.g. \citet{JS}, Section III.6). We are now in position to formulate the following classical definition.

\begin{definition}	\label{strategy}
Let $a\in\R_+$. An element $H\in L(S;\PP)$ is said to be an \emph{$a$-admissible strategy} if $H_0=0$ and 
$(H\cdot S)_t\geq -a$ $\PP$-a.s. for all $t\in[0,T]$. An element $H\in L(S;\PP)$ is said to be an  \emph{admissible strategy} if it is an $a$-admissible strategy for some $a\in\R_+$.
\end{definition}

For $a\in\R_+$, we denote by $\A_a(\PP)$ the set of all $a$-admissible strategies and by $\A(\PP)$ the set of all admissible strategies, i.e., $\A(\PP)=\bigcup_{a\in\R_+}\A_a(\PP)$. As usual, $H^i_t$ represents the number of units of asset $i$ held in the portfolio at time $t$. For $(x,H)\in\R\times\A(\PP)$, we define the \emph{gains from trading} process $G(H)=\bigl(G_t(H)\bigr)_{0\leq t \leq T}$ by $G_t(H):=(H\cdot S)_t$, for all $t\in[0,T]$, and the \emph{portfolio value} process $V(x,H):=x+G(H)$. This corresponds to consider portfolios generated by \emph{self-financing} admissible strategies.

We now introduce five different notions of arbitrage that will be considered in the present paper. 

\begin{definition}	\label{arb} 
\begin{itemize}
\item[(i)]
A strategy $H\in\A_0(\PP)$ is said to yield an \emph{increasing profit} if the process $G(H)$ is $\PP$-a.s. non-decreasing and $\PP\bigl(G_T(H)>0\bigr)>0$. If there exists no such strategy we say that the \emph{No Increasing Profit (NIP)} condition holds;
\item[(ii)]
a strategy $H\in\A_0(\PP)$ is said to yield a \emph{strong arbitrage opportunity} if $\PP\bigl(G_T(H)>0\bigr)>0$. If there exists no such strategy we say that the \emph{No Strong Arbitrage (NSA)} condition holds;
\item[(iii)]
a non-negative random variable $\xi$ is said to yield an \emph{arbitrage of the first kind} if $\PP(\xi>0)>0$ and for every $v\in(0,\infty)$ there exists an element $H^v\in\A_v(\PP)$ such that $V_T(v,H^v)\geq\xi$ $\PP$-a.s. If there exists no such random variable we say that the \emph{No Arbitrage of the First Kind (NA1)} condition holds;
\item[(iv)]
a strategy $H\in\A(\PP)$ is said to yield an \emph{arbitrage opportunity} if $G_T(H)\geq0$ $\PP$-a.s. and  $\PP\bigl(G_T(H)>0\bigr)>0$. If there exists no such strategy we say that the \emph{No Arbitrage (NA)} condition holds;
\item[(v)]
a sequence $\{H^n\}_{n\in\N}\subset\A(\PP)$ is said to yield a \emph{Free Lunch with Vanishing Risk} if there exist a positive constant $\varepsilon>0$ and an increasing sequence $\{\delta_n\}_{n\in\N}$ with $0\leq\delta_n\nearrow1$ such that $\PP\bigl(G_T(H^n)>-1+\delta_n\bigr)=1$ and $\PP\bigl(G_T(H^n)>\varepsilon\bigr)\geq\varepsilon$. If there exists no such sequence we say that the \emph{No Free Lunch with Vanishing Risk (NFLVR)} condition holds.
\end{itemize}
\end{definition}

The NIP condition corresponds to the notion of No Unbounded Increasing Profit introduced in \citet{KK} and represents the weakest condition among those listed above.\footnote{Note that if $H$ yields an increasing profit in the sense of Definition \ref{arb}-(i), it holds that $H^n:=nH\in\A_0(\PP)$ and $G(H^n)\geq G(H)$, for every $n\in\N$. This means that the increasing profit generated by $H$ can be scaled to arbitrarily large levels of wealth, thus explaining the adjective unbounded.} The NSA condition corresponds to the notion of absence of arbitrage opportunities adopted in Section 3 of \citet{LW} as well as to the NA$^+$ condition studied in \citet{Str}. The notion of arbitrage of the first kind goes back to \citet{Ing}, while the formulation adopted in Definition \ref{arb} is due to \citet{Ka1}. In particular, the NA1 condition is equivalent to the notion of No Unbounded Profit with Bounded Risk of \citet{KK} (see \citet{Ka1}, Proposition 1), which in turn corresponds to the BK condition considered in \citet{Kab}. Finally, the NA and NFLVR conditions are classical (see e.g. \citet{DS94}).

We close this section by recalling the probabilistic characterisations of the conditions introduced in Definition \ref{arb}, referring to \citet{F} for a more detailed analysis of the no-arbitrage properties of financial models based on continuous semimartingales.
As a preliminary, let the \emph{mean-variance trade-off process} $\hatK=(\hatK_t)_{0\leq t \leq T}$ be defined as (see e.g. \citet{Sch})
\begin{equation}	\label{MVT}
\hatK_t := \int_0^t\!\lambda^{\top}_ud\langle M,M\rangle_u\,\lambda_u
= \int_0^t\!a^{\top}_uc^+_ua_u\,dB_u,
\qquad \text{for all }t\in[0,T].
\end{equation}
Let us also define $\hatK_s^t:=\hatK_t-\hatK_s$, for all $s,t\in\left[0,T\right]$ with $s\leq t$, and the (possibly infinite valued) stopping time
\begin{equation}	\label{sigma}
\sigma:=\inf\bigl\{t\in[0,T]:\hatK_t^{t+h}=\infty,\forall h\in\left(0,T-t\right]\bigr\}\,,
\end{equation}
with the usual convention $\inf\emptyset=\infty$.
The next proposition provides necessary and sufficient conditions for the validity of the no-arbitrage conditions introduced in Definition \ref{arb}, collecting several important results obtained by \citet{DS94,DS95b}, \citet{Kab}, \citet{Str}, \citet{KK} and \citet{Ka1}. In the current formulation, a proof can be found in \citet{F}.

\begin{proposition}	\label{arb-char}
The following hold:
\begin{itemize}
\item[(i)] NIP holds if and only if $\nu=0$ $\PP\otimes B$-a.e.;
\item[(ii)] NSA holds if and only if $\nu=0$ $\PP\otimes B$-a.e. and $\sigma=\infty$ $\PP$-a.s.;
\item[(iii)] NA1 holds if and only if $\nu=0$ $\PP\otimes B$-a.e. and $\hatK_T<\infty$ $\PP$-a.s.;
\item[(iv)] NFLVR holds if and only if NA1 and NA hold;
\item[(v)] NFLVR holds if and only if there exists $\QQ\sim\PP$ such that $S\in\mMloc^c(\QQ)$;
\end{itemize}
where $\nu$, $\sigma$ and $\hatK$ are defined in \eqref{char-2}, \eqref{sigma} and \eqref{MVT}, respectively.
\end{proposition}

\section{Absolutely continuous changes of measure}	\label{S3}

Let $\QQ$ be a probability measure on $(\Omega,\F)$ with $\QQ\ll\PP$, but not necessarily equivalent to $\PP$. It is well known (see e.g. \citet{JS}, Section III.3) that there exists a unique (up to $\PP$- and $\QQ$-indistinguishability) non-negative $\PP$-martingale $Z=(Z_t)_{0\leq t \leq T}$ that satisfies $Z_t=d\QQ|_{\F_t}/d\PP|_{\F_t}$, for all $t\in[0,T]$ (note however that $Z$ is not necessarily continuous). Moreover, we also have $\QQ\bigl(Z_t>0\text{ and }Z_{t-}>0\text{ for all }t\in\left[0,T\right]\bigr)=1$. The Girsanov-Lenglart theory of absolutely continuous changes of measure (see \citet{Len})  allows us to compute the canonical decomposition of $S$ with respect to $\QQ$.

\begin{lemma}	\label{can-dec}
Let $\QQ$ be a probability measure on $(\Omega,\F)$ such that $\QQ\ll\PP$. The canonical decomposition of $S$ with respect to $\QQ$ is given by $S=S_0+\barA+\barM$, where
\begin{equation}	\label{can-dec-1}	\begin{aligned}
\barA &:= \int\!\left(c\,\bar{\lambda}+\bar{\nu}\right)dB\,,
\quad \text{with } \bar{\lambda}:=\lambda+\theta/Z_- \text{ and }\, \bar{\nu}:=\nu\,,	\\
\barM &:= M-\int\!\frac{c\,\theta}{Z_-}\,dB\,\in\,\mMloc^c(\QQ)\,,
\end{aligned}	\end{equation}
where the $\R^d$-valued predictable process $\theta=(\theta_t)_{0\leq t \leq T}\in L^2_{\emph{loc}}(M;\PP)$ can be chosen such that $\theta=c^+d\langle M,Z\rangle/dB$ $\PP\otimes B$-a.e.
\end{lemma}
\begin{proof}
Theorem 1 of \citet{Len} implies that $S$ is a continuous semimartingale with respect to $\QQ$ and, hence, it admits a unique canonical decomposition under $\QQ$. Since $S$ is continuous, the predictable quadratic variation $\langle M,Z\rangle$ always exists and the process $\barM:=M-\int\!Z^{-1}_-d\langle M,Z\rangle$ belongs to $\Mloc^c(\QQ)$, see Theorem 2 of \citet{Len}. By applying the Galtchouk-Kunita-Watanabe decomposition to $Z$ with respect to $M$ (see \citet{AS}), we get the existence of a process $\theta\in L^2_{\text{loc}}(M;\PP)$ such that $d\langle M,Z\rangle_t=d\langle M,M\rangle_t\theta_t=c_t\,\theta_t\,dB_t$. Together with \eqref{char}-\eqref{char-2}, this gives the canonical decomposition \eqref{can-dec-1} of $S$ with respect to the measure $\QQ$. 
Finally, define the process $\tilde{\theta}:=c^+c\,\,\theta$, so that $\tilde{\theta}=c^+d\langle M,Z\rangle/dB$ holds $\PP\otimes B$-a.e. Since $c^+$ is predictable and using the properties of the Moore-Penrose pseudoinverse, it is easy to check that $\tilde{\theta}$ belongs to $L^2_{\text{loc}}(M;\PP)$ and satisfies $c\,\tilde{\theta}=c\,\theta$.
\end{proof}

Similarly as in the previous section, let the process $\hatK^{\QQ}=(\hatK^{\QQ}_t)_{0\leq t \leq T}$ be defined by
\[
\hatK^{\QQ}_t =: \int_0^t\!\bar{\lambda}^{\top}_ud\langle\barM,\barM\rangle_u\,\bar{\lambda}_u\,,
\qquad\text{for all }t\in[0,T]\,,
\]
where $\bar{\lambda}$ and $\barM$ are as in Lemma \ref{can-dec}. Furthermore, let us define the (possibly infinite valued) stopping time $\sigma^{\QQ}$ by
\[
\sigma^{\QQ} := \inf\bigl\{t\in[0,T]:\hatK^{\QQ}_{t+h}-\hatK^{\QQ}_t=\infty,\forall h\in\left(0,T-t\right]\bigr\}\,.
\]

The next theorem is the main result of the first part of the present paper and shows that \emph{weak} no-arbitrage conditions (i.e., NIP, NSA and NA1), as opposed to \emph{strong} no-arbitrage conditions (i.e., NA and NFLVR), are always stable with respect to absolutely continuous changes of measure and not only to equivalent changes of measure. The proof is a rather direct consequence of Proposition \ref{arb-char} and Lemma \ref{can-dec}, but we prefer to give full details for the reader's convenience.

\begin{theorem}	\label{main}
Let $\QQ$ be a probability measure on $(\Omega,\F)$ such that $\QQ\ll\PP$. Then the following hold:
\begin{itemize}
\item[(i)] if NIP holds with respect to $\PP$, then NIP holds with respect to $\QQ$ as well;
\item[(ii)] if NSA holds with respect to $\PP$, then NSA holds with respect to $\QQ$ as well;
\item[(iii)] if NA1 holds with respect to $\PP$, then NA1 holds with respect to $\QQ$ as well.
\end{itemize}
\end{theorem}
\begin{proof}
If NIP holds with respect to $\PP$, Lemma \ref{can-dec} and part \emph{(i)} of Proposition \ref{arb-char} imply that $\bar{\nu}=\nu=0$ $\PP\otimes B$-a.e. and also $\QQ\otimes B$-a.e., since $\QQ\ll\PP$. The first claim then follows again from part \emph{(i)} of Proposition \ref{arb-char} (now applied under $\QQ$). In order to prove parts \emph{(ii)}-\emph{(iii)}, observe that, for every $t\in[0,T]$:
\begin{equation}	\label{main-1}
\hatK^{\QQ}_t = \int_0^t\!\left(\lambda_u+\frac{\theta_u}{Z_{u-}}\right)^{\!\!\top}\!
c_u\left(\lambda_u+\frac{\theta_u}{Z_{u-}}\right)dB_u
\leq 2\,\hatK_t + 2\int_0^t\!\frac{1}{Z_{u-}^2}\,\theta^{\top}_uc_u\,\theta_u\,dB_u\,,
\end{equation}
due to the Kunita-Watanabe inequality together with \eqref{MVT}.
Note that, due to Lemma \ref{can-dec}, we have $\theta\in L^2_{\text{loc}}(M;\PP)\subseteq L^2_{\text{loc}}(\barM;\QQ)$, where the last inclusion follows from the assumption that $\QQ\ll\PP$. Since the process $Z_-$ is $\QQ$-a.s. strictly positive, adapted and left-continuous, hence locally bounded, the second term on the right-hand side of \eqref{main-1} is $\QQ$-a.s. finite for every $t\in[0,T]$. In particular, this implies that $\sigma^{\QQ}\geq\sigma$ $\QQ$-a.s. By part \emph{(ii)} of Proposition \ref{arb-char}, if NSA holds with respect to $\PP$, then $\sigma=\infty$ $\PP$-a.s. and also $\QQ$-a.s., since $\QQ\ll\PP$, thus giving $\sigma^{\QQ}=\infty$ $\QQ$-a.s. Then, again part \emph{(ii)} of Proposition \ref{arb-char} implies that NSA holds with respect to $\QQ$. Similarly, due to part \emph{(iii)} of Proposition \ref{arb-char}, if NA1 holds with respect to $\PP$ then $\hatK_T<\infty$ $\PP$-a.s. and also $\QQ$-a.s., since $\QQ\ll\PP$. By \eqref{main-1}, this implies that $\hatK^{\QQ}_T<\infty$ $\QQ$-a.s. and, hence, again by part \emph{(iii)} of Proposition \ref{arb-char}, NA1 holds with respect to $\QQ$.
\end{proof}

\begin{remark}
The condition $\hatK_T<\infty$ $\PP$-a.s., which characterises NA1, is called \emph{finiteness condition} in \citet{DShi}. In Lemma 4.5 of that paper, it is shown that the finiteness condition is stable under equivalent changes of measure. A result analogous to part \emph{(iii)} of Theorem \ref{main} (albeit with a different terminology) is also established in Proposition 2.7 and Remark 2.10 of \citet{CS}. 
\end{remark}

We close this section by pointing out that, if we restrict our attention to \emph{equivalent} changes of measure, rather than only absolutely continuous changes of measure, then also the NA and NFLVR conditions are stable. Indeed, suppose that NA holds with respect to $\PP$ and let $\QQ\sim\PP$. Arguing by contradiction, suppose that there exists a strategy $H\in\A(\QQ)$ which realises an arbitrage opportunity under $\QQ$, in the sense of part (iv) of Definition \ref{arb}. Then, since $\PP\ll\QQ$, Proposition III.6.24 of \citet{JS} shows that $H\in L(S;\PP)$ and the stochastic integral $H\cdot S$ viewed with respect to $\PP$ coincides with the stochastic integral with respect to $\QQ$. Moreover, since $\QQ\sim\PP$, we also have $H\in\A(\PP)$ as well as
$G_T\left(H\right)\geq 0$ $\PP$-a.s. and $\PP\bigl(G_T\left(H\right)>0\bigr)>0$, thus contradicting the validity of NA with respect to $\PP$. An analogous reasoning allows to show the stability of NFLVR with respect to equivalent changes of measure.
Note also that this argument does not rely on the continuity of $S$.

\section{The NA and NFLVR conditions: a counterexample}	\label{S4}

As shown in Theorem \ref{main}, in the context of continuous semimartingale models, the \emph{weak} NIP, NSA and NA1 conditions are always stable with respect to absolutely continuous changes of measure. We now show that, in general, the classical \emph{strong} NA and NFLVR conditions are not robust with respect to absolutely continuous changes of measure, even for continuous processes. We proceed to illustrate this phenomenon by means of a counterexample, which has been already developed in \citet{DS95a} in relation to strict local martingales (i.e., local martingales which fail to be martingales, see e.g. \citet{ELY}).

Let $W=(W_t)_{0\leq t \leq T}$ be a real-valued Brownian motion starting from $W_0=1$ and the stopping time $\tau:=\inf\left\{t\in[0,T]:W_t=0\right\}\wedge T$. We define $S$ as the stopped process $S:=W^{\tau}$ and let the filtration $\FF$ be the $\PP$-augmented natural filtration of $S$, with $\F=\F_T$. Since $S\in\M(\PP)$, NFLVR trivially holds with respect to $\PP$, in view of part \emph{(v)} of Proposition \ref{arb-char}. We then define a probability measure $\QQ\ll\PP$ by $d\QQ/d\PP:=S_T=W_{\tau}$. Clearly, $\PP\ll\QQ$ does not hold, since $\PP(S_T=0)>0$, so that $\QQ$ and $\PP$ fail to be equivalent. Note also that $S$ represents the density process of $\QQ$ with respect to $\PP$.

\begin{proposition}	\label{example}
In the context of the present section, the process $S$ allows for arbitrage opportunities with respect to the probability measure $\QQ$ in the filtration $\FF$.
\end{proposition}
\begin{proof}
By Theorem 2 of \citet{Len}, the process $N:=S-\int\!S^{-1}d\langle S\rangle$ belongs to $\Mloc^c(\QQ)$. Observe that $\langle N\rangle_t=\langle S\rangle_t=t\wedge\tau$ for all $t\in\left[0,T\right]$. Noting that $\QQ\left(S_T=0\right)=\EE_{\PP}\left[\ind_{\left\{S_T=0\right\}}S_T\right]=0$, it holds that $\tau=T$ $\QQ$-a.s. and hence $\langle N\rangle_t=t$ $\QQ$-a.s. for all $t\in\left[0,T\right]$. L\'evy's characterisation theorem then implies that $N$ is a $\QQ$-Brownian motion starting at $N_0=1$. 
Let us now denote by $\GG$ the $\QQ$-augmented natural filtration of $N$ (or, equivalently, of $S$), which coincides with $\FF$ augmented by the subsets of $\left\{S_T=0\right\}$. The process $S$ satisfies $dS_t=S_t^{-1}dt+dN_t$. Hence, with respect to the measure $\QQ$ and the filtration $\GG$, the process $S$ is a three-dimensional Bessel process (see e.g. \citet{RY}, Section XI.1) and the corollary on page 361 of \citet{DS95a} implies that $S$ admits arbitrage opportunities with respect to $\QQ$ in the filtration $\GG$.
This means that there exists a $\GG$-predictable admissible strategy $H\in L(S;\QQ)$ such that $G_T(H)\geq 0$ $\QQ$-a.s. and $\QQ\bigl(G_T(H)>0\bigr)>0$. As remarked on page 360 of \citet{DS95a}, there also exists an $\FF$-predictable process $K$ that is $\QQ$-indistinguishable from $H$, so that $K\cdot S\equiv H\cdot S$, where both stochastic integrals are considered with respect to $(\QQ,\GG)$. But then, since $S$ is $\FF$-adapted, Proposition III.6.25 of \citet{JS} shows that the stochastic integral $K\cdot S$ viewed with respect to $(\QQ,\GG)$ coincides with the stochastic integral viewed with respect to $(\QQ,\FF)$. We have thus proved that $K$ realises an arbitrage opportunity with respect to $\QQ$ in the filtration $\FF$.
\end{proof}

Actually, the counterexample given in Proposition \ref{example} can be generalised to a whole class of models for which NA (and, hence, NFLVR as well) is destroyed by an absolutely continuous, but not equivalent, change of measure. The following result corresponds essentially to Theorem 3 of \citet{DS95a} (see also Proposition 2.8 of \citet{OR} as well as Theorem 1 of \citet{RR} for an extension to incomplete markets).

\begin{proposition}	\label{example-gen}
Let $S=(S_t)_{0\leq t\leq T}$ be a non-negative real-valued process in $\M(\PP)$ with the predictable representation property with respect to $\PP$, with $S_0=1$ $\PP$-a.s. Define the probability measure $\QQ\ll\PP$ by $d\QQ/d\PP:=S_T$. If $\PP(S_T=0)>0$ then $S$ does not satisfy NA with respect to $\QQ$.
\end{proposition}

\begin{remark}
We want to point out that the results of the present section do not actually need the completeness of the filtration. Indeed, the predictable representation property can be also established in right-continuous but not complete filtrations, see e.g. Theorem III.4.33 of \citet{JS}, while the failure of NA under $\QQ$ in the proof of Proposition \ref{example} can be proven as on page 59 of \citet{FR}.
\end{remark}

\section{Stability of NA1 in general semimartingale models}	\label{S5}

Among the different notions of no-arbitrage considered in Definition \ref{arb}, the NA1 condition is of particular interest. Indeed, as shown in Proposition 4.19 of \citet{KK} (see also \citet{CDM}), NA1 is the minimal condition that allows for a meaningful solution of portfolio optimisation problems in general semimartingale models. This section aims at investigating whether the NA1 condition is stable with respect to absolutely continuous changes of measure in financial models based on a \emph{general} (i.e., not necessarily continuous or locally bounded) $\R^d$-valued semimartingale $S=(S_t)_{0\leq t\leq T}$.\footnote{As already explained in Remark 4.13 in \citet{KK}, the NA1 condition may no longer be preserved by an absolutely continuous change of measure if jumps are present, unlike the continuous case considered so far.}

The following result, due to \citet{T}, characterises NA1 in a general semimartingale setting. 
We recall that an $\R^d$-valued semimartingale $S$ is called a \emph{$\sigma$-martingale} if there exists an increasing sequence $\{\Sigma_n\}_{n\in\N}$ of predictable sets with $\bigcup_{n\in\N}\Sigma_n=\Omega\times[0,T]$ and such that $\ind_{\Sigma_n}\cdot S^i\in\M(\PP)$, for all $n\in\N$ and $i=1,\ldots,d$ (see \citet{JS}, Section III.6e).
We denote by $\M_{\sigma}(\PP)$ the family of all $\sigma$-martingales with respect to $\PP$.

\begin{proposition}	\label{NA1}
NA1 holds if and only if there exists a \emph{strict martingale density (SMD)}, i.e., a strictly positive real-valued process $L=(L_t)_{0\leq t\leq T}$ belonging to $\mMloc(\PP)$ such that $\EE_{\PP}[L_0]<\infty$ and $LS\in\M_{\sigma}(\PP)$.
\end{proposition}

Let now $\QQ$ be a probability measure on $(\Omega,\F)$ with $\QQ\ll\PP$. Until the end of this section, we denote by $Z=(Z_t)_{0\leq t \leq T}$ the density process of $\QQ$ with respect to $\PP$ and  define the following stopping times:
\begin{equation}	\label{st}	\begin{aligned}
\tau &:= \inf\{t\in[0,T]:Z_{t-}=0\text{ or }Z_t=0\}\,,\\
\tau_n &:= \inf\{t\in[0,T]:Z_t<1/n\}\wedge T\,, \qquad\text{for all }n\in\N.
\end{aligned}	\end{equation}
A key insight, first exploited by \citet{Fol} and \citet{Mey} and more recently by \citet{CFR} and \citet{KKN} among others, consists in looking at the process $1/Z$ under the measure $\QQ$.

\begin{lemma}	\label{properties-Z}
Let $\QQ$ be a probability measure on $(\Omega,\F)$ such that $\QQ\ll\PP$. 
Then the following hold:
\begin{itemize}
\item[(i)] the process $1/Z$ is a strictly positive $\QQ$-supermartingale;
\item[(ii)] for every $M=(M_t)_{0\leq t \leq T}\in\M(\PP)$ the process $M/Z$ belongs to $\mMloc(\QQ)$ if and only if $\PP(\tau>\tau_n)=1$, for all $n\in\N$.
\end{itemize}
\end{lemma}
\begin{proof}
The first claim follows from simple computations (see also \citet{CFR}, Theorem 2.1). In order to prove the second claim, suppose that $\PP(\tau>\tau_n)=1$, for all $n\in\N$. Then, for every $M\in\M(\PP)$, we have $M/Z\in\Mloc(\QQ)$, see e.g. part \emph{(iii)} of Proposition 2.3 of \citet{CFR}. Conversely, taking $M\equiv 1$, if $1/Z\in\Mloc(\QQ)$, Theorem 2.1 of \citet{CFR} implies that $\PP(\tau>\tau_n)=1$ for all $n\in\N$.
\end{proof}

The next theorem is the main result of this section and shows that the NA1 condition is stable with respect to an absolutely continuous change of measure if the corresponding density process $Z$ does not jump to zero. Note also that the proof is constructive, in the sense that it exhibits an explicit SMD for $S$ with respect to the measure $\QQ$.

\begin{theorem}	\label{main-2}
Let $\QQ$ be a probability measure on $(\Omega,\F)$ such that $\QQ\ll\PP$.
If NA1 holds with respect to $\PP$ and $\PP(\tau>\tau_n)=1$, for all $n\in\N$, then NA1 holds with respect to $\QQ$.
\end{theorem}
\begin{proof}
By Proposition \ref{NA1}, there exists a real-valued strictly positive process $L\in\Mloc(\PP)$ such that $\EE_{\PP}[L_0]<\infty$ and $LS\in\M_{\sigma}(\PP)$. Let $\{\varrho_n\}_{n\in\N}$ be a $\PP$-localizing sequence for $L$. Then, due to part \emph{(ii)} of Lemma \ref{properties-Z}, it holds that $L^{\varrho_n}/Z\in\Mloc(\QQ)$. Since $\QQ\ll\PP$, this implies that $L/Z\in\Mloc(\QQ)$. 
By Definition III.6.33 of \citet{JS}, there exists an increasing sequence of predictable sets $\{\Sigma_n\}_{n\in\N}$ such that $\bigcup_{n\in\N}\Sigma_n=\Omega\times[0,T]$ and $\ind_{\Sigma_n}\cdot LS^i\in\M(\PP)$, for every $i=1,\ldots,d$. Again by part \emph{(ii)} of Lemma \ref{properties-Z}, we have $(\ind_{\Sigma_n}\cdot LS^i)/Z\in\Mloc(\QQ)$ and an application of the product rule leads to $\ind_{\Sigma_n}\cdot(LS^i/Z)\in\Mloc(\QQ)$, for every $i=1,\ldots,d$. Since $\QQ\ll\PP$ and the class $\M_{\sigma}(\QQ)$ is stable by localization (see e.g. \citet{JS}, Proposition III.6.34), this shows that $LS^i/Z\in\M_{\sigma}(\QQ)$, for all $i=1,\ldots,d$. 
Since we also have $\EE_{\QQ}[L_0/Z_0]\leq\EE_{\PP}[L_0]<\infty$, we have thus shown that $L/Z$ is an SMD under $\QQ$. Proposition \ref{NA1} then implies that NA1 holds with respect to $\QQ$.
\end{proof}

In general, the condition $\PP(\tau>\tau_n)=1$, for all $n\in\N$, cannot be weakened, as shown in the following result (see also Example \ref{ex-jump} and Proposition \ref{counter} below).

\begin{proposition}	\label{CE}
Let $S=(S_t)_{0\leq t \leq T}$ be a non-negative real-valued process in $\M(\PP)$ with the predictable representation property with respect to $\PP$, with $S_0=1$ $\PP$-a.s. Define the probability measure $\QQ\ll\PP$ by $d\QQ/d\PP:=S_T$ and $\tau:=\inf\{t\in[0,T]:S_{t-}=0\text{ or }S_t=0\}$. 
Then, if $\PP\bigl(\{\tau\leq T\}\cap\{S_{\tau-}>0\}\bigr)>0$, the process $S$ does not satisfy NA1 with respect to $\QQ$.
\end{proposition}
\begin{proof}
Arguing by contradiction, suppose that $S$ satisfies NA1 under $\QQ$. Then, by Proposition \ref{NA1}, there exists an SMD $L$ with respect to $\QQ$, so that $LS\in\Mloc(\QQ)$, since a non-negative $\sigma$-martingale is a local martingale (see e.g. \citet{JS}, Proposition III.6.35). Due to the lemma on page 67 of \citet{Len}, there exists a sequence of stopping times $\{\varrho_n\}_{n\in\N}$ with $\varrho_n\nearrow\infty$ $\QQ$-a.s. as $n\rightarrow\infty$ such that $(LS)^{\varrho_n}\in\Mloc(\PP)$ and $(LS^2)^{\varrho_n}\in\Mloc(\PP)$. In view of Theorem 11.2 of \citet{J}, since $S\in\M(\PP)$ enjoys the predictable representation property with respect to $\PP$, this implies that $(LS)^{\varrho_n}$ is $\PP$-a.s. trivial, for all $n\in\N$. Since $\QQ\ll\PP$ and $\varrho_n\nearrow\infty$ $\QQ$-a.s. as $n\rightarrow\infty$, this implies that $L=1/S$ $\QQ$-a.s. However, due to part \emph{(ii)} of Lemma \ref{properties-Z} (see also \citet{PR}, Example 4.1), the condition $\PP\bigl(\{\tau\leq T\}\cap\{S_{\tau-}>0\}\bigr)>0$ implies that $1/S$ is a $\QQ$-supermartingale which does not belong to $\Mloc(\QQ)$, thus contradicting the hypothesis that $L$ is an SMD with respect to $\QQ$.
\end{proof}

\begin{example}	\label{ex-jump}
We now present a simple example where the assumptions of Proposition \ref{CE} are satisfied (compare also with \citet{PR}, Example 4.2).
Let $(\Omega,\F,\PP)$ be a given probability space supporting a standard exponential random variable $\xi:\Omega\rightarrow\R_+$, so that $\PP(\xi>t)=e^{-t}$, for all $t\in\R_+$, and let $\FF$ be the right-continuous filtration generated by the process $(\ind_{\{\xi\leq t\}})_{0\leq t \leq T}$. Define then the process $S=(S_t)_{0\leq t\leq T}$ by $S_t:=\ind_{\{\xi>t\}}e^t$, for $t\in[0,T]$. As follows from Propositions 7.2.3.2 and 7.2.5.1 of \citet{JYC}, the process $S$ belongs to $\M(\PP)$ and enjoys the predictable representation property in $\FF$. Clearly, using the notation of Proposition \ref{CE}, it holds that $\tau=\xi$ and, moreover:
\[
\PP\bigl(\{\tau\leq T\}\cap\{S_{\tau-}>0\}\bigr)
= \PP(\xi\leq T) = 1-e^{-T}>0.
\]
Proposition \ref{CE} then shows that $S$ fails to satisfy NA1 under $\QQ$. Indeed, since $\QQ(\xi>T)=1$, the process $S$ is perceived as a strictly increasing process under $\QQ$ and, therefore, it allows for increasing profits (in the sense of part (i) of Definition \ref{arb}).
\end{example}

It is interesting to compare Proposition \ref{CE} with the result of Proposition \ref{example-gen}. Indeed, both propositions show that, if one starts from a model where $S\in\M(\PP)$ (and, hence, NFLVR trivially holds with respect to $\PP$) and $S$ enjoys the predictable representation property, then arbitrage profits may arise after an absolutely continuous, but not equivalent, change of measure from $\PP$ to $\QQ$. More specifically, Proposition \ref{example-gen} shows that, if $d\QQ=S_T\,d\PP$ fails to define a measure $\QQ\sim\PP$, then the process $S$ allows for arbitrage opportunities (i.e., NA fails) when viewed under $\QQ$. Moreover, due to Proposition \ref{CE}, if the process $S$ can jump to zero, then $S$ also allows for arbitrages of the first kind when viewed under $\QQ$ (compare also with Example \ref{ex-jump}).

\begin{remark}
In the recent paper \citet{RR}, the authors provide a systematic procedure for constructing market models that satisfy NA1 but allow for arbitrage opportunities (i.e., NA fails to hold). Their approach is intimately linked to our results: indeed, they start from a market model where $S$ is a non-negative $\R^d$-valued martingale and then pass to an absolutely continuous, but not equivalent, probability measure whose density process is allowed to reach zero only continuously.
Hence, the change of measure adopted in \citet{RR} has the effect of disrupting the NA component of NFLVR, while, in view of Theorem \ref{main-2}, the NA1 component is preserved.
\end{remark}

The next proposition represents a converse result to Theorem \ref{main-2} and shows that, if the condition $\PP(\tau>\tau_n)=1$ for all $n\in\N$ fails to hold, then one can find a semimartingale $S$ satisfying NFLVR with respect to $\PP$ but allowing for increasing profits (and, hence, arbitrages of the first kind) under $\QQ$. The proof exploits an idea already used in the proof of Theorem 3.16 of \citet{CADJ}.

\begin{proposition}	\label{counter}
Let $\QQ$ be a probability measure on $(\Omega,\F)$ such that $\QQ\ll\PP$ and suppose that $\PP(\tau=\tau_n)>0$ for some $n\in\N$. Then there exists a process $S\in\M(\PP)$ which allows for increasing profits with respect to $\QQ$.
\end{proposition}
\begin{proof}
Let $B=(\ind_{\lsi\tau,T\rsi})^p$ be the predictable compensator (see e.g. \citet{JS}, Theorem I.3.17) of the increasing process $\ind_{\lsi\tau,T\rsi}$ and define $S:=-(\ind_{\lsi\tau,T\rsi}-B)$, so that $S\in\M(\PP)$. Since $\tau=\infty$ $\QQ$-a.s., it holds that $S=B$ $\QQ$-a.s. The predictable process $H:=\ind_{\lsir0,T\rsi}$ satisfies $G(H)=H\cdot S=B-B_0\geq0$ $\QQ$-a.s., since $B$ is $\PP$-a.s. increasing and $\QQ\ll\PP$. Moreover, using the properties of predictable compensators (see \citet{J}, Proposition 1.47) together with the $\PP$-martingale property of $Z$, we get
\[
\EE_{\QQ}\bigl[G_T(H)\bigr]
= \EE\bigl[Z_TB_T-Z_0B_0\bigr]
= \EE\left[\int_0^T\!Z_{t-}\,dB_t\right]
= \EE\left[\int_0^T\!Z_{t-}\,d\ind_{\lsi\tau,T\rsi}\right]
= \EE\left[Z_{\tau-}\ind_{\{\tau\leq T\}}\right]>0,
\]
where the last inequality follows from the fact that $\PP(\tau=\tau_n)>0$ for some $n\in\N$, meaning that the process $Z$ has a strictly positive probability of jumping to zero.
\end{proof}

Summing up, in the context of general $\R^d$-valued semimartingale models, Lemma \ref{properties-Z}, Theorem \ref{main-2} and Proposition \ref{CE} together yield the equivalence between the three following statements, where $\QQ\ll\PP$ and according to the notation introduced in \eqref{st}:
\begin{itemize}
\item[(a)] \emph{for every $\R^d$-valued semimartingale $S=(S_t)_{0\leq t\leq T}$ satisfying NA1 with respect to $\PP$, NA1 holds with respect to $\QQ$ as well};
\item[(b)] \emph{$\PP(\tau>\tau_n)=1$, for all $n\in\N$};
\item[(c)] $1/Z\in\Mloc(\QQ)$.
\end{itemize}

\vspace{0.5cm}

\begin{spacing}{1}
\subsubsection*{Acknowledgements}
This research was supported by a Marie Curie Intra European Fellowship within the 7th European Community Framework Programme under grant agreement PIEF-GA-2012-332345. 
The author is thankful to Johannes Ruf for useful comments on an earlier version of the paper.
\end{spacing}

\begin{spacing}{1.1}

\end{spacing}

\end{document}